\documentclass[a4paper,11pt]{article}

\usepackage{amsthm,amsmath,amssymb,amsfonts,pstricks,pst-node,graphics,graphicx}

\def\im{\operatorname{Im}}

\newcommand\Ord{{\rm Ord}}
\newcommand\cS{{\mathcal S}}

\def\fb{{c}}

\def\hb{{b}}

\newcommand\cK{{\mathrm k}}

\newcommand\cF{{\mathrm F}}

\newcommand\cR{{\mathcal R}}

\newcommand\cQ{{\mathcal Q}}
\newcommand\cA{{\mathcal A}}

\def\bbbc{{\mathbb C}}
\def\bbbz{{\mathbb Z}}

\def\bbbq{{\mathbb Q}}

\newtheorem{Rem}{Remark}

\newtheorem{Def}{Definition}
\newtheorem{The}{Theorem}
\newtheorem{Pro}{Proposition}
\newtheorem{Lem}{Lemma}
\newtheorem{Ex}{Example}
\newtheorem{Cor}{Corollary}

\begin{document}

\bibliographystyle{unsrt}
\title{PreHamiltonian and Hamiltonian operators for  differential-difference 
equations}
\author{Sylvain Carpentier $^\ddagger$, Alexander V. Mikhailov$^{\star}$ and 
Jing Ping 
Wang$ ^\dagger $
\\
$\ddagger$ Mathematics Department, Columbia University, USA
\\
$\dagger$ School of Mathematics, Statistics \& Actuarial Science, University of 
Kent, UK \\
$\star$ Applied Mathematics Department, University of Leeds, UK
}
\date{}  

\maketitle
\begin{abstract}
In this paper we are developing a theory of rational (pseudo) difference 
Hamiltonian operators, focusing in particular 
on its algebraic aspects. We show that a 
pseudo--difference Hamiltonian operator can be represented as a ratio $AB^{-1}$ 
of two difference operators  with coefficients from a difference field $\cF$ 
where $A$ is preHamiltonian. A difference operator $A$ is called 
preHamiltonian if its image is a Lie subalgebra with respect to the 
Lie bracket of evolutionary vector fields on $\cF$. We show that a skew-symmetric difference operator is 
Hamiltonian if and only if it is preHamiltonian and satisfies simply verifiable 
conditions on its coefficients. We 
show that if $H$ is a rational Hamiltonian operator, then to find a second 
Hamiltonian operator $K$ compatible with $H$ is the same as to find a 
preHamiltonian pair $A$ and $B$ such that $AB^{-1}H$ is skew-symmetric. We 
apply our theory to non-trivial multi-Hamiltonian structures of 
Narita-Itoh-Bogayavlensky and Adler-Postnikov equations.
\end{abstract}

\section{Introduction}

Poisson brackets play a fundamental role in the study of Hamiltonian systems of 
ordinary differential equations. The same holds true for  their infinite 
dimensional analogues, this is to say systems of partial 
differential, or differential--difference, equations. They are particularly important in the theory of 
integrable systems and in deformation quantisation. Zakharov and Faddeev have 
shown that the Korteweg--de-Vries equation can be viewed as a completely 
integrable Hamiltonian system for a Poisson bracket (called the 
Gardner-Zakharov-Faddeev bracket) defined in terms of the Hamiltonian differential operator 
$\frac{d}{dx}$. The concept of Hamiltonian pairs was introduced by  Magri 
\cite{Mag78}. 
Equations which admit two compatible Hamiltonian structures are called 
bi--Hamiltonian. 
\\
In almost all known to us examples of scalar 
differential--difference 
bi-Hamiltonian equations at least 
one of the Hamiltonian operators is  rational. Only the Volterra chain possesses two difference 
Hamiltonian operators (see Examples 1 and 2 in Section 3.1 of this paper). This 
justifies the necessity  to develop a rigorous theory of 
rational Hamiltonian and recursion operators. In our paper \cite{CMW18a} we have 
extended the results obtained in the differential setting \cite{Carp2017} to 
the difference case. In particular, we have shown that rational recursion 
operators generating the symmetries of an integrable differential-difference 
equation must be factorisable as a ratio of two compatible preHamiltonian 
difference operators.  In this paper 
we develop the theory of the so-called preHamiltonian operators, and study their interrelations with rational Hamiltonian operators. We will  illustrate 
our results using Adler-Postnikov integrable
differential-difference equation for which the Hamiltonian structure was not 
know previously \cite{Adler2}.

Let us consider the well-known modified Volterra chain \cite{hirota730, Yami}
\begin{eqnarray}\label{mvol}
u_t=u^2 (u_1-u_{-1}) ,
\end{eqnarray}
where $u$ is a function of a lattice variable $n\in \bbbz$ and continuous time variable $t$. Here
we use the shorthand notations
\begin{eqnarray*}
 u_t=\partial_t(u), \quad u_j=\cS^j u(n,t)=u(n+j,t)
\end{eqnarray*}
and $\cS$ is the shift operator. The right hand side of  the equation \eqref{mvol} lies in the 
difference field $\cF=\bbbc(...,u_{-1},u,u_1,...)$  of rational functions in 
the generators $u_i$, $i \in \bbbz$. 
It possesses a rational recursion operator 
$$R=u^2 \cS+2uu_1+u^2\cS^{-1}+2 u^2(u_1-u_{-1})(\cS-1)^{-1}\frac{1}{u} 
$$
with \((\cS-1)^{-1} \) standing for the  inverse of \(\cS-1\). 
Thus this recursion operator is only defined on \(u\im (\cS-1)\).
The operator $R$ can be written as a ratio of two difference operators
\begin{equation}\label{volab}
R=A B^{-1}, \  \mbox{where $A=u^2 (\cS-\cS^{-1}) u(\cS+1)$ and $B=u (\cS-1 )$}. 
\end{equation}
The pair of difference operator $A$ and $B$ \textit{generates} the hierarchy 
of 
symmetries of the modified Volterra chain. We have shown in \cite{CMW18a} that 
the difference operators $A$ and $B$ must then form a preHamiltonian pair, that is, 
any linear combination $C=A+\lambda B$, $\lambda\in \cK$ satisfies 
\begin{equation*}
 [\im \, C, \im \, C] \subseteq \im \, C,
 \end{equation*}
 where the Lie bracket on $\cF$ is given by $[a,b]=X_a(b)-X_b(a)$ for $a,b \in 
\cF$ and   $X_a=\sum_{n \in 
\bbbz}{S^n(a)\frac{\partial }{\partial u_n}}$ is the evolutionary derivation of 
the field $\cF$ with characteristic function $a$.

 The recursion operator $R$
can also be presented as $R=H_2 H_1^{-1}$ with
\begin{eqnarray}\label{mvalham}
&&H_1=u (\cS-1) (\cS+1)^{-1} u \quad \mbox{and} \quad  H_2=u^2 (\cS-\cS^{-1}) u^2 .
\end{eqnarray}
The operators $H_1$ and $H_2$ are Hamiltonian operators. The difference operator 
$H_2$ is preHamiltonian and  induces a Poisson bracket $\{f,g\}_2=\smallint 
\delta_u (f)H_2\delta_u(g)$  on the space of functionals  
$f,g\in\cF'=\cF/(\cS-1)\cF$, where $\delta_u$ denotes the variational derivative 
with respect to the dependent variable $u$
\[
\delta_u (a)=\sum_{n\in\bbbz}\cS^{-n}\frac{\partial a}{\partial u_n},\qquad 
a\in\cF. 
\]
 The  Hamiltonian operator $H_1$ is 
rational, it can be represented as $H_1=\hat{A}\hat{B}^{-1}$ where $\hat{A}=u 
(\cS-1)$ is preHamiltonian and  $\hat{B}=\frac{1}{u}(\cS+1)$. It induces a 
Poisson bracket on a smaller space $F_{\hat{B}}'=\{  f \in \cF' | \delta_u f \in 
Im \hat{B} \}$. The modified Volterra chain (\ref{mvol}) is a bi-Hamiltonian system for the pair of compatible Hamiltonian operators  $H_1,H_2$
\[
u_t=H_1\ \delta_u u u_1=H_2\ \delta_u \ln u .
\] 
It follows from Theorem 4 in Section 3.2 that the sequence $R^nH$, $n \in 
\mathbb{Z}$, form a
family of compatible  rational
Hamiltonian operators for the system (\ref{mvol}).

The arrangement of this paper is as follows:  In Section \ref{Sec2}, we give a 
short description of
the skew field of rational (pseudo--)difference operators, i.e. operators  
of the form $AB^{-1}$, where $A$ and $B$ are difference operators. We recall the algebraic properties of the noncommutative ring of 
difference operators.
In particular, it is a principal ideal domain, which is right and left 
Euclidean  and satisfies the right (left) Ore property. 
We then define the Fr{\'e}chet derivative of rational operators and introduce 
the notion of bi-difference operators.

The main results are presented in Section \ref{Sec3}. We  explore the 
interrelations between preHamiltonian and Hamiltonian operators. We first look at the difference (pre)Hamiltonian operators in Section \ref{sec31}.
We prove that a Hamiltonian operator is a skew-symmetric preHamiltonian 
operator with simple  conditions on its coefficients:
\begin{quote}
{\it A difference operator $H=\sum_{i=1}^N{h^{(i)}S^{i}-S^{-i}h^{(i)}}$ is Hamiltonian $\iff$ $H$ is preHamiltonian and $h^{(i)}=h^{(i)}(u,...,u_i)$ for all $i=1,...,N$.}
\end{quote}
In Section \ref{sec32} we then generalise the definition of Hamiltonian 
difference operators to  rational (pseudo difference) operators  and 
demonstrate that preHamiltonian pairs provide us with a method to find 
compatible Hamiltonian rational operators to a given (rational) Hamiltonian 
operator. We have shown that
\begin{quote}
{\it A rational Hamiltonian operator and a skew-symmetric rational operator $K$ form a Hamiltonian pair if and 
only if there exists a preHamiltonian pair  of difference operators  $A$ and $B$ 
such that $HK^{-1}=AB^{-1}$.
}
\end{quote}

In Section \ref{Sec5}, 
we apply the theoretical results  to a new integrable equation 
derived by Adler and Postnikov  
\cite{Adler2}:
\begin{equation}
 \label{adpost}
u_t=u^2(u_2u_1-u_{-1}u_{-2})-u(u_1-u_{-1}).
\end{equation}
We show that the equation (\ref{adpost}) is a Hamiltonian system 
$$
u_t=H \delta_u \ln u
$$
with the rational Hamiltonian  operator
\begin{equation*} 
\begin{split}
H&=u^2u_1u_2^2\cS^2-\cS^{-2}u^2u_1u_2^2+\cS^{-1}uu_1(u+u_1)-uu_1(u+u_1)\cS \\
&+u(1-\cS^{-1})(1-uu_1)({\cS}u-u\cS^{-1})^{-1}(1-uu_1)(\cS-1)u\, .
\end{split}
\end{equation*}
In \cite{CMW18a} we have found a rational Nijenhuis recursion operator $R$ for 
the equation (\ref{adpost}). We show that the sequence $R^nH, n \in \bbbz$ forms a
family of compatible  rational
Hamiltonian operators for the system (\ref{adpost}).

\section{Difference and Rational difference operators}\label{Sec2}
In this section, we briefly recall some notations and statements that were 
introduced and discussed in detail in Section 2 of our paper 
\cite{CMW18a}. 
In the end of this section, we prove two lemmas on (bi)difference operators, 
which we are going to use in the next section.
Although we only consider the scalar case in the following, most of our results 
can be generalised to rational matrix operators.

Let $\cK $ be a zero characteristic ground field,  such as 
$\bbbc$  or $\bbbq$. We define the polynomial ring 
$$\mathrm{K}=\cK[\ldots, u _{-1}, u _0, u _1,\ldots]$$  in the infinite set of 
variables 
$\{u\}=\{ u _k;\ k\in\bbbz\}$ and the corresponding field of fractions 
$$\cF=\cK(\ldots, u _{-1}, u _0, u _1,\ldots).$$ Note that every element 
of  
$\mathrm{K}$ and $\cF$ depends on a finite number of variables only. 

There is a natural automorphism $\cS$ of the field $\cF$, 
which we call the shift operator, defined as 
\[
 \cS: a( u _k,\ldots , u _r)\mapsto a( u _{k+1},\ldots , u _{r+1}),\quad 
\cS:\alpha\mapsto\alpha, \qquad a( u _k,\ldots , u _r)\in \cF,\ \ \alpha\in\cK.
\]
We often use the shorthand notation 
$
 a_i=\cS^i(a)=a( u _{k+i},\ldots , u _{r+i}),
i\in\bbbz ,
$
and omit the index zero in $a_0$ or $u_0$ when there is no ambiguity.
The field $\cF$  equipped with  the  automorphism $\cS$ is a 
difference field  and 
the ground field $\cK$ is its subfield  of constants.

The partial derivatives $\frac{\partial}{\partial u_i},\ i\in\bbbz$ are 
commuting derivations of $\cF$ satisfying the conditions
\begin{equation}\label{spart}
\cS \frac{\partial}{\partial u_{i}}= \frac{\partial}{\partial u_{i+1}} \cS.
\end{equation}

A derivation of $\cF$ is said to be evolutionary if it commutes with
the shift operator $\cS$. Such a derivation is completely determined 
by one element of $a \in \cF$ and is of the form
\begin{equation}
 \label{Xf}
 X_f=\sum_{i\in\bbbz}\cS^i(a) \frac{\partial}{\partial u_{i}},\qquad a \in\cF.
\end{equation}
The element $a$ is called the characteristic of the evolutionary derivation $X_a$. 
The action of $X_a(ab)$ for  $b\in\cF$ can also be represented in the form
\[
 X_a(b)=b_*[a],
\]
where $b_*[a]$ is the Fr\'echet derivative of $b=b(u_p,\ldots,u_q)$ in the 
direction $a$, which is defined as 
\[
 b_*[a]:= \frac{d}{d\epsilon}b(u_p+\epsilon f_p,\ldots,u_q+\epsilon 
f_q)\arrowvert_{\epsilon=0}=\sum_{i=p}^q\frac{\partial b}{\partial u_i}a_i.
\]
The Fr\'echet derivative of $b=b(u_p,\ldots,u_q)$ is
a difference operator represented by a 
finite sum
\begin{equation}\label{astar}
b_*=\sum_{i=p}^q\frac{\partial b}{\partial u_i}\cS^i.
\end{equation}

Evolutionary derivations form a Lie $\cK$-subalgebra $\cA$ in the the 
 Lie algebra of derivations of the field
$\cF$. Indeed, 
\[\begin{array}{l}
   \alpha X_a+\beta X_b=X_{\alpha a+\beta b},\qquad \alpha,\beta\in\cK,\\
\phantom{}   [X_a,X_b]=X_{[a,b]},
  \end{array}
\]
where $[a,b]\in\cF$ denotes the Lie bracket
\begin{equation}\label{bracket}
 [a,b]=X_a(b)-X_b(a)=b_*[a]-a_*[b]. 
\end{equation}
The bracket (\ref{bracket}) is $\cK$--bilinear, antisymmetric and satisfies 
the Jacobi identity. Thus $\cF$, equipped with the bracket (\ref{bracket}), has 
a structure of  a Lie algebra over $\cK$.

\begin{Def}\label{deford} A difference operator $B$ of order ${\rm ord}\,  
B:=(M,N)$ with 
coefficients in
$\cF$ is a finite sum of the form
\begin{equation}\label{operB}
B= b^{(N)} \cS^{N}+b^{(N-1)} \cS^{N-1}+\cdots +b^{(M)} \cS^{M},\qquad b^{(k)}\in\cF, \ \ M\le 
N,\
\ N,M\in\mathbb{Z}
\end{equation}
where $b_N$ and $b_M$ are non-zero. The term $b^{(N)} \cS^{N}$ is called the leading monomial of $B$.
The total order of $B$ is defined as ${\rm Ord}B=N-M$. The total order of 
the zero 
operator is defined as $\Ord\, 0:=\{\infty\}$.
\end{Def}

The Fr\'echet derivative (\ref{astar}) is an example of a difference operator 
of order $(p,q)$ and total order ${\rm Ord}\, a_* =q-p$. For an element  
$a \in \cF$ the order and total 
order are defined as ${\rm ord}\, a_*$ and   ${\rm Ord}\, a_*$ respectively.

Difference operators form a unital ring $\cR=\cF[\cS,\cS^{-1}]$ of 
Laurent 
polynomials in $\cS$ with coefficients in $\cF$, equipped with the usual addition 
and multiplication defined by 
\begin{equation} \label{smult}
 a\cS^n \cdot b\cS^m=a\cS^n(b)\cS^{n+m}, \hspace{4 mm} a,b \in \cF,  \hspace{4 mm} n,m \in \bbbz.
\end{equation}
This multiplication is associative, but non-commutative. 
The ring $\cR$ is a right and left Euclidean domain and it
satisfies the right (left) Ore property, that is, 
for any $A,B\in \cR$ their exist $A_1,B_1$, not both equal to 
zero, such that $AB_1=BA_1$, (resp. $B_1 A=A_1 B$). In other words, the right (left) ideal
$A \cR \cap B \cR$ (resp. $ \cR A \cap  \cR B$) is nontrivial. Its generator  $M$ has total order
$\Ord A+ \Ord B-\Ord D$, where $D$ is the greatest left (resp. right) common 
divisor of $A$ and $B$.
The domain $\cR$ can be naturally embedded in the skew field of 
rational pseudo--difference operators.

\begin{Def}
 A rational pseudo--difference operator is an element of
 \[
  \cQ=\{ AB^{-1}\,|\, A,B\in \cR,\ B\ne0\}.
 \]
 We shall call them rational operator for simplicity.
\end{Def}

Any rational operator 
$L=AB^{-1}$ can also be written in the form $L=\hat{B}^{-1}\hat A,\ \hat A,\hat 
B\in\cR$ and $\hat B\ne 0$. 
Thus any statement for the representation $L=AB^{-1}$ can be easily 
reformulated to the representation $L=\hat B^{-1}\hat A$. 
In particular, we have shown in \cite{CMW18a} that rational operators
 $\cQ$ form a skew field with respect to usual addition and multiplication.
  The decomposition $L=AB^{-1},\ A,B\in \cR$ of an element $L\in 
\cQ$ is unique if we require that $B$ has a minimal possible total order 
with leading monomial being $1$.

The definition of the total order for difference operators (Definition 
\ref{deford}) can be extended to rational operators:
\begin{equation*}\label{ordF}
 \Ord\, (AB^{-1}):=\Ord\,A-\Ord\, B,\quad A,B \in \cR. 
\end{equation*}

\begin{Def}
A formal adjoint operator $A^\dagger$ for any $A\in\cQ$ can be 
defined 
recursively by the rules:  $a^\dagger=a$ for any $a\in\cF$, 
$\cS^\dagger=\cS^{-1},\ (A+B)^\dagger=A^\dagger+B^\dagger$ 
 and $(A\cdot B)^\dagger=B^\dagger \cdot 
A^\dagger$ for any $A,B\in\cQ$.
In particular $(A^{-1})^\dagger=(A^\dagger)^{-1}$ and 
$(a\cS^n)^\dagger=\cS^{-n}a=a_{-n}\cS^{-n}$. 

A rational operator $A\in \cQ$ is called  
skew-symmetric if $A^{\dagger}=-A$.
\end{Def}

While difference operators act naturally on elements of the field $\cF$,   
rational operators cannot be a priori applied to elements of $\cF$. 
Similarly to the theory of rational differential operators \cite{CDSK2015} for  
$L=AB^{-1} \in \cQ$ and $a,b \in \cF$ we define the {\sl correspondance} $a=Lb$ 
when there exists $c \in \cF$ such that $a=Ac$ and $b=Bc$.

Finally we define the Fr\'echet derivative of difference operators and rational difference operators.
\begin{Def}
The Fr\'echet derivative of a difference operator $B$ (\ref{operB}) in the direction of $a \in \cF$ is defined as
\begin{equation}\label{freopd}
B_*[a]= b^{(N)}_*[a] \cS^{N}+b^{(N-1)}_*[a] \cS^{N-1}+\cdots +b^{(M)}_*[a] \cS^{M} .
\end{equation}
\end{Def}
Here we can also view $B_*$ as a bidifference operator in the sense that, for a 
given $a \in \cF$, both $B_*[\bullet](a)$ and $B_*[a]$ 
are in $\cR$, i.e. difference operators on $\cF$. 
For convenience, we introduce the notation 
 $D_B$ as the following bidifference operator:
\begin{equation}\label{fredop}
(D_B)_a(b)=B_*[b](a) \quad \mbox{for all} \quad a,b \in \cF.
\end{equation}
This definition  can be naturally extended to rational operators: 
$(AB^{-1})_*=A_*B^{-1}-AB^{-1}B_*B^{-1}$.

We complete this  section by proving two lemmas on (bi)difference operators, 
which we are going to use in section $3$. For a bidifference operator $M$ and an element $a \in \cF$ we denote the 
difference operator $M(a, \bullet)$ by $M_a$.
\begin{Lem}\label{CD}
Let $C$ and $D$ be  two difference operators and  $P, Q$ be two bidifference 
operators on $\cF$ such that $CP_a=Q_aD$ for all $a 
\in \cF$. 
Then there exists a bidifference operator $M$ such that $P_a=M_aD$ for all $a \in \cF$. 
\end{Lem}
\begin{proof}
There exist two bidifference operators $M$ and $R$ such that 
$$P_a=M_aD+R_a\quad  \mbox{and}\quad \Ord \, R_a< \Ord \, D\quad  \mbox{for all $a \in \cF$}.$$ 
We know that $CP_a=Q_aD$, that is to see $CR_a=N_aD$, where $N_a=Q_a-CM_a$ for all $a \in \cF$. Let us assume that $R_a \neq 0$. There exist difference operators $R^j$ and $N^i$ such that for all $a \in \cF$,
\begin{equation}
R_a =\sum_{j=l}^k {a_j R^j}, \hspace{2 mm}
N_a =\sum_{i=m}^n {a_i N^i}
\end{equation}
In particular $\Ord R^j  < \Ord D$ for all $j=l,...,k$. If $fS^r$ is the leading 
term of $C$ we must have 
\begin{equation}
N^n D=fS^rR^k,
\end{equation}
which implies that $\Ord R^k \geq \Ord D$ contradicting to $\Ord \, R_a< \Ord \, D$.
\end{proof}

\begin{Lem}\label{CDE}
Let $C$, $D$ and $E$ be  non-zero difference operators such that $C+\lambda D$ 
divides $E$ on the right for all $\lambda \in 
\cK$. Then there exists $a \in \cF$  and a difference operator $X$ such that $XC=aXD$ and $E=XD$.
\end{Lem}
\begin{proof}
We first prove the statement  in the case where $\Ord \,C=\Ord \,D=0$. Since $C$ 
and $D$ are invertible difference operators, we can assume that $C=1$ and 
$D=bS^n$. We want to show that if a difference operator $E$ is divisible on the 
right by $1+\lambda b \cS^n$ for all $\lambda \in \cK$, then $n=0$. Assume that 
$n \neq 0$ and define the difference operator $M_{\lambda}$ for $\lambda \in 
\cK$ uniquely by 
\begin{equation} \label{auxe1}
E=M_{\lambda}(1+ \lambda b\cS^n).
\end{equation}
 It is clear since $n \neq 0$ that the coefficients of $M_{\lambda}$ are elements of $\cF[\lambda, \lambda^{-1}]$. In other words, $M_{\lambda}$ is an element of $\cR[\lambda, \lambda^{-1}]$. We get a contradiction looking at \eqref{auxe1} in $\cR[\lambda, \lambda^{-1}]$ since we assumed $E \neq 0$. Hence $n=0$.
\\
We now prove the Lemma in the general case by induction on $\Ord \, E$. If $\Ord \,E=0$ then $\Ord \, (C + \lambda D)=0$ for all $\lambda \in \cK$, which implies that $\Ord \, C= \Ord \, D=0$, which we have treated already. Assume then that $\Ord \, E >0$ and that $C + \lambda D$ divides $E$ on the right for all $ \lambda \in \cK$. Let $MC=ND$ be the left least common multiple (llcm) of $C$ and $D$. Both $C$ and $D$ divide $E$ on the right, hence so does their llcm. Therefore there exists a difference operator $G$ such that $E=GMC=GND$. As earlier we define for all $\lambda \in \cK$ the operator $M_{\lambda}$ by 
\begin{equation} \label{auxe2}
E=M_{\lambda}(C+ \lambda D).
\end{equation}
Substituing $E=GMC$ in \eqref{auxe2} and using the definition of the llcm there exist $P_{\lambda} \in \cR$ for all $\lambda \in \cK$ such that 
\begin{equation} \label{auxe3}
 GM-M_{\lambda}=P_{\lambda}M; \quad
\lambda M_{\lambda}=P_{\lambda}N.
\end{equation}
Similarly there exist $Q_{\lambda} \in \cR$ for all $\lambda \in \cK$ such that 
\begin{equation} \label{auxe4}
 GN-\lambda M_{\lambda}=Q_{\lambda}N; \quad
   M_{\lambda}=Q_{\lambda}M.
\end{equation}
From \eqref{auxe3} and \eqref{auxe4} we can see that for all $\lambda \in \cK$, $G=P_{\lambda}+Q_{\lambda}$  and 
\begin{equation} \label{auxe5}
 GN=Q_{\lambda}(N+\lambda M); \quad
\lambda GM=P_{\lambda}(N+\lambda M).
\end{equation}
If $\Ord \,C=\Ord \, D=0$ we have nothing to prove. Otherwise, without loss of generality, we can assume that $\Ord \,D >0$. Hence $\Ord \, GN <  \Ord \, E$. We see from \eqref{auxe5} that $N+ \lambda M$ divides $GN$ on the right for all $\lambda \in \cK$. By the induction hypothesis, one can find a difference operator $Y$ and an element $a \in \cF$ such that $YN=aYM$ and $GN=YM$. Let $X=YM$. We have $XC=YMC=YND=aYMD=aXD$, which concludes the proof.
\end{proof}

\section{PreHamiltonian and Hamiltonian operators}\label{Sec3}

In  this section, we start by recalling the definitions of preHamiltonian and 
Hamiltonian difference operators. We explain how they relate to each other and 
introduce the class of rational Hamiltonian operators. In particular, we prove 
that given a (rational) Hamiltonian operator $H$, to find a Hamiltonian operator 
compatible to $H$ is the same as to find a preHamiltonian pair $A$ and $B$ 
such that the operator $AB^{-1}H$ is skew-symmetric.

\subsection{Definitions and Interrelations with Examples}\label{sec31}
\begin{Def}
 A difference operator $A$ is called preHamiltonian if $\im A$ is a Lie 
subalgebra, i.e.,
 \begin{equation}\label{preH}
 [\im A,\ \im A]\subseteq \im A 
 \end{equation}
\end{Def}
By direct computation, it is easy to see that  a difference operator $A$ is 
preHamiltonian if and only if there 
exists a 2-form on $\cF$ denoted by $\omega_{A}$ such that (c.f. 
\cite{mr1923781})
\begin{equation}\label{PreH1}
 A_*[A a](b)-A_*[A b](a)=A (\omega_{A}(a,b))\quad \mbox{ for all }\quad a,b \in 
\cF ,
\end{equation}
where $A_*$ denotes the Fr{\'e}chet derivative of the operator $A$.
More precisely, $\omega_A$ is a bidifference operator, i.e. $\omega_A(a,b)$ is 
a combination of terms of the form
$ca_ib_j$ where $c \in \cF$ and $i,j \in \bbbz$.  
Using the notation introduced in (\ref{fredop}), the identity (\ref{PreH1}) is equivalent to 
\begin{equation} \label{preH2}
A_*[Aa]-(D_A)_a A=A \omega_{A}(a, \bullet) \hspace{3 mm} \text{for all} \hspace{3 mm} a \in \cF.
\end{equation}

The preHamiltonian operator $A$ defines a Lie 
algebra on $\cF/\ker A$ with the Lie bracket
\[
 A([a,b]_A)=[Aa, Ab].
\]
The bracket $[a,b]_A$ is anti-symmetric, $\cK$--linear and satisfies the Jacobi 
identity. The latter follows from the fact that $A(\cF)$ is a Lie subalgebra 
with respect to the standard Lie bracket  (\ref{bracket}).

We can construct higher order preHamiltonian operators from known ones 
using the following two lemmas. The first one appeared in \cite{mr1923781} in 
the context of scalar preHamiltonian differential operators of arbitrary 
order.
\begin{Lem}\label{lbn}
Assume that $A$ is a preHamiltonian difference operator. For any difference 
operator $C$, the operator $AC$ is preHamiltonian if and only if
\begin{eqnarray*}
 \delta(a, b)=\omega_A(Ca,Cb)+C_*[AC a](b)-C_*[AC b](a)
\end{eqnarray*}
is in the image of $C$ for all $a,b \in \cF$.
\end{Lem}
\begin{proof}
According to (\ref{PreH1}), we compute,  for all $a, b \in \cF$,
\begin{eqnarray*}
  (AC)_*[AC a](b)-(AC)_*[AC b](a)=A \left(\omega_A(Ca,Cb)+C_*[AC a](b)-C_*[AC b](a)\right).
\end{eqnarray*}
Therefore, we only need to check whether $\delta(a, b)$ is in the image of the operator $C$.
\end{proof}
\begin{Rem} \label{inv}
If $A$ is a preHamiltonian operator with associated form  $\omega_A$ and $Q$ is 
a invertible difference operator then $B=AQ$ is also preHamiltonian and the 
previous Lemma provides us with an explicit formula for $\omega_B$
\begin{eqnarray*}
\omega_B(a,b)=Q^{-1}(\omega_A(Qa,Qb)+Q_*[B a](b)-Q_*[B b](a)).
\end{eqnarray*}
\end{Rem}
\begin{Lem}\label{rlcmpreh}
If $A$ and $B$ are preHamiltonian difference operators,  then their right least 
common multiple is also preHamiltonian.
\end{Lem}
\begin{proof} Let $M=AD=BC$ be the right least common multiple (rlcm) of $A$ and $B$.
Then $[\im M,\ \im M]=[\im AD,\ \im AD]\subseteq [\im A,\  \im A] \subseteq \im 
A$ since $A$ is a preHamiltonian operator. Similarly, 
$[\im M,\ \im M] \subseteq \im B$. Moreover we have  $\im M=  \im A \cap \im B$ 
(Lemma $10$ in \cite{CMW18a}). Therefore, $[\im M,\ \im M] \subseteq \im M$.
\end{proof}

Similarly to Hamiltonian operators, in general, a
linear combination of two preHamiltonian operators is no longer preHamiltonian. 
This naturally leads to the following 
definition:
\begin{Def}\label{defpair}
 We say that two difference operators $A$ and $B$ form a preHamiltonian pair if 
$A+\lambda B$ is 
preHamiltonian for all constant $\lambda \in \cK$.
\end{Def}
A preHamiltonian pair $A$ and $B$ implies the existence of 2-forms $\omega_A$, 
$\omega_B$ and $\omega_{A+\lambda 
B}=\omega_A +\lambda \omega_B$. They satisfy
\begin{eqnarray}\label{pair}
A_*[B a](b) +B_*[A a](b)-A_*[B b](a)-B_*[A b](a)=A \omega_B(a,b)+B 
\omega_A(a,b) 
 \  \mbox{ for all }\  a,b \in 
\cF .
\end{eqnarray}
Using the notation introduced by (\ref{fredop}), equation \eqref{pair} is equivalent to
\begin{eqnarray}\label{paireq}
A_*[B a] +B_*[A a]-(D_A)_a B -(D_B)_a A =A \omega_B(a,\bullet)+B \omega_A(a,\bullet) 
 \  \mbox{ for all }\  a\in \cF . 
\end{eqnarray}
\begin{Pro}\label{cpair}
Let $A$ and  $B$ be a preHamiltonian pair.  If there exists an operator $C$ such 
that $AC$ and $BC$ are both preHamiltonian, then they again form a 
preHamiltonian pair.
\end{Pro}
\begin{proof}
Let $\omega_A$ and $\omega_B$ be the $2$-form  associated to preHamiltonian 
operators $A$ and $B$, that is,
\begin{eqnarray*}
A_*[Aa]&=(D_A)_aA+A \omega_A(a, \bullet), \qquad B_*[Ba]&=(D_B)_a B+B \omega_B(a, \bullet)
\end{eqnarray*}
for all $a \in \cF$. The forms $\omega_A$ and  $\omega_B$ satisfy (\ref{pair}) 
since $A$ and $B$ form a Hamiltonian pair.
According to Lemma \ref{lbn}, we know that there  exist two bidifference 
operators $M$ and $N$ such that for all $a, b \in \cF$
\begin{eqnarray*}
&&\omega_A(Ca,Cb)+C_*[AC a](b)-C_*[AC b](a)=C M(a, b);\\
&&\omega_B(Ca,Cb)+C_*[BC a](b)-C_*[BC b](a)=C N(a, b).
\end{eqnarray*}
Substituting them into (\ref{pair}) for $Ca$ and $Cb$, we get 
\begin{equation*}
(AC)_*[BC a](b) +(BC)_*[AC a](b)-(AC)_*[BC b](a)-(BC)_*[AC b](a)=AC N(a,b)+BC M(a,b),
\end{equation*}
which implies that $AC$ and $BC$ for a preHamiltonian pair.
\end{proof}

Before we move on to justify the terminology  \textit{preHamiltonian}, we first 
recall the definition of a Hamiltonian difference operator. 

For any element $a\in \cF$, we define an equivalent class (or a functional) 
$\int\! a$
by saying that two elements $a,b\in\cF$ are equivalent if \(a-b\in
\mbox{Im}(\cS-1)\).  The space of functionals is denoted by $\cF'$.
For any functional $\int\!\!f\in \cF'$ (simply written $f\in \cF'$ without 
confusion), we define its difference variational derivative (Euler operator) 
denoted by
$\delta_{u} f \in \cF$ (here we identify the dual space with itself) as
$$\delta_{u} f=\sum_{i\in\bbbz} \cS^{-i}  \frac{\partial f}{\partial u_{i}}=
\frac{\partial }{\partial u}\left(\sum_{i\in\bbbz} \cS^{-i}  f \right).$$
\begin{Def} \label{locham}
A difference operator $H$ is Hamiltonian if the bracket 
\begin{equation} \label{lieham}
\{ f, g \}_H := \smallint \delta_u f \cdot H(\delta_u g)
\end{equation}
defines a Lie bracket on $\cF'$.
\end{Def}

As in the differential case \cite{mr94j:58081} this definition can be 
re-cast  purely in terms of operators acting on the difference field $\cF$ and 
avoiding computations on the quotient space $\cF'$ of functionals. 
\begin{The}\label{algham}
A difference operator $H$ is Hamiltonian if and only if $H$ is skew-symmetric 
and  
\begin{eqnarray}\label{lochbis}
 H_*[Ha]-(D_H)_aH=H {(D_H)_a}^{\dagger} \text{  for all } a \in \cF,
\end{eqnarray}
where ${(D_H)_a}^{\dagger}$ is the adjoint operator of $(D_H)_a$ defined in (\ref{fredop}).
\end{The}
\begin{proof}
We first prove the following: if $a \in \cF$ is such that $\int a\cdot \delta_u f =0$ for all $ f \in \cF'$, then $a=0$.
Since $(\cS-1) \cF \subset \ker  \, \delta_u$, we have $\delta_u (a\cdot \delta_u f)=0$ for all $ f \in \cF'$. In particular we can consider $f=uu_k$ for $k \in \bbbz$. Let $(m,p)$ be the order of $a$. We have for all $k \geq 0$,
\begin{equation}\label{eqa}
a_k+a_{-k}+\sum_{n=m}^p{(u_{k-n}+u_{-k-n})S^{-n}(\frac{\partial a}{\partial u_n})}=0.
\end{equation}
For a given $n$ and for $k$ large enough in \eqref{eqa}, after applying $\frac{\partial}{\partial u_{n-k}}$ we get
\begin{equation} \label{eqb}
S^k(\frac{\partial a}{\partial u_{-n}})=S^{-n}(\frac{\partial a}{\partial u_n})
\end{equation}
Since \eqref{eqb} holds for all $k$ large enough, we deduce that $\frac{\partial a}{\partial u_n}=0$. Hence $a=0$.

The anti-symmetry of \eqref{lieham} is equivalent to the skew-symmetry of the operator $H$. Indeed, \eqref{lieham} is anti-symmetric if and only if 
\begin{equation}\label{eqc}
\smallint \delta_u f \cdot (H+H^{\dagger})(\delta_u g)=0 \text{  for all  }  f,g \in \cF'.
\end{equation}
From what we just proved, this is equivalent to say that $(H+H^{\dagger})(\delta_u f)=0$ for all $ f \in \cF'$, hence that $H+H^{\dagger}=0$ since nonzero difference operators have finite dimensional kernel 
over the constants.

Finally we look at the Jacobi identity. For this, we take $a=\delta_u f, b=\delta_u g$ and $c=\delta_u h$, where $f, g, h \in \cF'$. Note that
\begin{equation}\label{jac1}
\begin{split}
\{  f, \{  g,  h\}_H \}_H&=\smallint  a\cdot H(\delta_u (  b\cdot H c))
  =-\smallint H a\cdot \delta_u ( b\cdot H c)
  =\smallint H a\cdot \delta_u (  c\cdot H b)\\
  &=\smallint ( c \cdot  H b)_*[H a]
 =\smallint   c \cdot (H b)_*[H a]+ \smallint H b\cdot  c_*[H a].
\end{split}
\end{equation}
Similarly,  we have 
\begin{equation}\label{jac2}
\begin{split}
\{  g, \{  h,  f\}_H \}_H&=\!\!-\{  g, \{ f,  h\}_H \}_H=\!\!-\smallint   c\cdot 
(H  a)_*[H b]- \smallint H a\cdot  c_*[H b] .
\end{split}
\end{equation}
As for the third term, we simply write
\begin{equation}\label{jac3}
\{  h, \{ f,  g\}_H \}_H=\smallint  c\cdot H(\delta_u ( a\cdot H b)).
\end{equation}
Since $c_*=c_*^\dagger$, this leads to
\begin{eqnarray*}
&&\{  f, \{  g,  h\}_H \}_H+\{  g, \{  h,  f\}_H \}_H+\{  h, \{ f,  g\}_H \}_H\nonumber\\
&=&\smallint   c \cdot \left((H b)_*[H a]- (H  a)_*[H b]+H (a_*^{\dagger}(Hb)+ (H b)_*^{\dagger}(a) ) \right)=0, 
\end{eqnarray*}
which itself is equivalent to
\begin{equation}\label{jaceq}
[Hb, Ha]=(H  a)_*[H b]-(H b)_*[H a]=H (a_*^{\dagger}(Hb)+ (H b)_*^{\dagger}(a)).
\end{equation}
Using the notation introduced in (\ref{fredop}), we have $(Hb)_*=(D_H)_b+Hb_*$ for all $b \in \cF$, which leads to
$(Hb)_*^{\dagger}=(D_H)_b^{\dagger}+b_*^{\dagger}H^{\dagger}=(D_H)_b^{\dagger}-b_*^{\dagger}H$. 
Since $a_*$ and $b_*$ are self-adjoint, we can write 
$$a_*^{\dagger}(Hb)+ (H b)_*^{\dagger}(a)=a_*(Hb)-b_*(Ha)+(D_H)_b^{\dagger}(a).$$
Moreover, $(Ha)_*[Hb]=H_*[Hb](a)+H(a_*[Hb])=(D_H)_aHb+H(a_*[Hb])$. Therefore from \eqref{jaceq} we deduce that \eqref{lochbis} holds on $\delta_u \cF \times \delta_u \cF$. 
We proved that equation \eqref{jaceq} holds for any $(a,b) \in \cF \times \cF$ since it is enough to check that it holds for any $(a,b) \in V \times V$, where $V$ is a subspace of $\cF$ 
infinite dimensional over the constants, and $V=\delta_u \cF$ provides us with such a subspace. 
\end{proof}

This theorem immediately implies that a Hamiltonian operator $H$ is preHamiltonian with 
\begin{eqnarray}\label{omegaH}
\omega_H(a, b)=(D_H)_a^{\dagger}(b). 
\end{eqnarray}
Note that the skew-symmetry of operator $H$ is a necessary condition since $\omega_H$ is a $2$-form.
This can be used as a criteria to determine whether an operator is Hamiltonian. 
Using formula (\ref{omegaH}), we have the following result for scalar difference operators:
\begin{The} \label{charham}
A skew-symmetric operator $H=\sum_{i=1}^k \left(h^{(i)} \cS^i-\cS^{-i} h^{(i)}\right)$ of total order $2k$  $(k>0)$ is Hamiltonian if and only if it is preHamiltonian and its coefficients $h^{(i)}$ only depend on $u,...,u_{i}$ for all $i=1,...,k$.
\end{The}
\begin{proof}
 First we assume that $H$ is a Hamiltonian operator, and show that its coefficients $h^{(i)}$ only depend on $u,...,u_i$. 
It follows that $H$ satisfies (\ref{lochbis}), that is,
\begin{eqnarray}\label{locheq2}
 H_*[Ha]=(D_H)_aH+H {(D_H)_a}^{\dagger} \text{  for all } a \in \cF.
\end{eqnarray}
This identity is an equality between bidifference operators, that is between summands of the form $ba_n\cS^m$ for $b \in \cF$ and $n,m \in \bbbz$.
The left hand side of \eqref{locheq2} is a difference operator in $\cS$ of order $(-k,k)$, or in other words a sum of terms of the form $ba_n\cS^m$ with $|m| \leq k$.  Hence so must be the right hand side of \eqref{locheq2} (RHS).  We can rewrite the RHS as 
\begin{equation}\label{locheq3}
\sum_{i=1}^k{(a_i-a_{-i}\cS^{-i}){h^{(i)}}_*H}-\sum_{i=1}^k{H{h^{(i)}}_*^{\dagger}(a\cS^{i}-a_i)}
\end{equation}
In the second term of \eqref{locheq3}, it is clear that every summand $ba_n\cS^m$ is such that $|m-n| \leq k$. 
Combining this remark with the fact that as a difference operator in $\cS$ \eqref{locheq3} has order $(-k,k)$, we deduce that any subterm $ba_n\cS^m$ appearing in the first term of \eqref{locheq3} must be such that $|m| \leq k$ or $|m-n| \leq k$. Therefore, given $i$ such that $1 \leq i \leq k$, as a difference operator in $\cS$, $a_{-i}S^{-i}{h^{(i)}}_*H$ cannot involve powers of $\cS$ below $S^{-i-k}$. This implies that ${h^{(i)}}_*$ does not depend on negative powers of $\cS$ (recall that $H$ has order $(-k,k)$). Similarly, the operator $a_i{h^{(i)}}_*H$ cannot involve powers of $\cS$ strictly bigger than $\cS^{k+i}$, which implies that ${h^{(i)}}_*$ can only depend on $1,..., \cS^{i}$. 

Conversely, we need to show that a skew-symmetric preHamiltonian operator $H$ such that all its coefficients $h^{(i)}$ depend only on $u, ...,u_i$ is Hamiltonian. For any $a \in \cF$, we write
$$P_a=H_*[Ha]-(D_H)_aH-H{(D_H)_a}^{\dagger}.$$ 
We want to prove that $P_a$ is identically $0$. 
Under the assumption, we have that $P_a$ is skew-symmetric and its total order is at most $4k$. 
We also know since $H$ is preHamiltonian that $H$ divides $P_a$ on the left for all $a \in \cF$. Of course $H$ must also divide $P_a$ on the right since $P_a$ and $H$ are both skew-symmetric. 
Therefore by Lemma \ref{CD} there exists $Q$ bidifference operator such that $P_a=HQ_aH$ for all $a \in \cF$. Moreover, $Q_a$ is skew-symmetric, hence its total order is at least $2$ if it is non-zero. 
Therefore $Q=0$.
\end{proof}
A recent  classification of low order scalar Hamiltonian 
operators in the   framework of multiplicative Poisson 
$\lambda$-brackets \cite{DKVW1}   is consistent with 
this theorem.
\begin{Ex}\label{ex1}
Consider the well-known Hamiltonian operator $H=u(\cS-\cS^{-1})u=u u_1 \cS-\cS^{-1} u u_1$ of the Volterra equation $u_t=u (u_1-u_{-1})$.
Obviously, $H$ is skew-symmetric and its  coefficient $h^{(1)}=uu_1$  only depends on $ u, u_1$. To conclude that it is indeed Hamiltonian using Theorem \ref{charham},
one needs to check that $H$ is preHamiltonian. Indeed, for all $a, b \in \cF$:
\begin{equation}
H_*[H a](b)-H_*[H b](a)=H\left( \frac{1}{u}(b H(a)-aH(b))\right).
\end{equation}
\end{Ex}

\begin{Ex}\label{ex2}
We can do the same for the second Hamiltonian operator of the Volterra equation 
\begin{eqnarray*}
&&K=u(\cS u\cS+u\cS+\cS u-u\cS^{-1}-\cS^{-1}u-\cS^{-1}u\cS^{-1})u\\
&&\quad= u u_1 u_2\cS^2+(u^2u_1+u u_1^2)\cS-\cS^{-1} (u^2u_1+u u_1^2)-\cS^{-2} u u_1 u_2 .
\end{eqnarray*}
Note that it is skew-symmetric and its coefficients $h^{(1)}=u^2u_1+u u_1^2$ depending on $u, u_1$ and $h^{(2)}=u u_1 u_2$ depending on $u, u_1 u_2$. 
To check that $K$ is preHamiltonian, we denote $A=K \frac{1}{u}$ and it follows from
\begin{equation*}
A_*[A a](b)-A_*[A b](a)=A(u(a_1b_{-1}+a_1b+ab_{-1}-a_{-1}b-ab_1-a_{-1}b_1)) \quad \mbox{for all} \quad a,b \in \cF.
\end{equation*}
\end{Ex}

In the same manner, we can use (\ref{omegaH}) to 
determine a Hamiltonian pair. 
The operators $H$ and $K$ form a Hamiltonian pair if and only if $$\omega_H(a, b)=(D_H)_a^{\dagger}(b),
\omega_K(a, b)=(D_H)_a^{\dagger}(b)\mbox{ and $\omega_{H+\lambda K}(a, b)=(D_{H+\lambda K})_a^{\dagger}(b)$
for all $a, b\in \cF$.}$$
Moreover, we are able to prove the statement on the relation between perHamiltonian and Hamiltonian pairs.
\begin{The}\label{repair}
Let $A$ and  $B$ be a preHamiltonian pair.  Assume that there exists a 
difference operator $C$ such that $AC$ is 
skew-symmetric and $BC$ is Hamiltonian. Then $AC$ is also Hamiltonian and
forms a Hamiltonian pair with $BC$.
\end{The}

In the next section we shall give a more general result in Theorem \ref{thham} 
and the proof of the above theorem will be a simple Corollary.
A special case of Theorem \ref{repair} is when the operator $C=1$, which leads to the following result.
\begin{Cor}
Let $A$ and  $B$ be a preHamiltonian pair such that $A$ is skew-symmetric and $B$ is Hamiltonian. Then $A$ is also Hamiltonian and
forms a Hamiltonian pair with $B$.
\end{Cor}
\begin{Ex} Consider the Volterra chain $u_t=u (u_1-u_{-1})$. It possesses a recursion operator 
$$
 R=A B^{-1}, \  \mbox{where $A=u(\cS +1)(u \cS-\cS^{-1} u)$, $B=u (\cS-1)$},
$$
and $A, B$ form a preHamiltonian pair.
Take $C=(1+\cS^{-1})u$. In Example \ref{ex1}, we verified that $BC$ is Hamiltonian. Notice that
$AC$ is skew-symmetric. Using the above theorem, we obtain that it is a 
Hamiltonian operator and  forms a Hamiltonian pair with $BC$.
\end{Ex}

\subsection{Generalisation to rational difference operators}\label{sec32}

In Examples 1-3 we illustrated our theory using the Hamiltonian 
structure of the Volterra hierarchy. Actually, the Voltera equation is the only 
example known to us of a scalar nonlinear difference equation possessing a 
compatible pair of difference  Hamiltonian operators\footnote{In the recent 
preprint \cite{DKVW1} the authors classified difference Hamiltonian operators of 
total order less or equal to $10$. However, it turned out 
that for the Hamiltonian pairs appeared in this classification, the 
hierarchies 
obtained through the Lenard scheme techniques were all equivalent to 
the Volterra chain.}. For all other integrable differential-difference equations known 
to us at least one Hamiltonian is a rational
(pseudo-difference) operator. In this Section we give all required definitions, develop 
the theory of rational Hamiltonian operators and study their relations with 
pairs of preHamiltonian difference operators.

Let $H$ be a skew-symmetric operator with decomposition $H=AB^{-1}$. It 
is defined on be the following subspace of $\cF'$ denoted by  $\cF'_B$, that is,
\begin{equation}
\cF'_B =\{  f \in \cF' | \delta_u f \in \im B\}.
\end{equation}
Note that if a difference operator $C$ divides $B$ on the left, then $\cF'_B \subseteq \cF'_C$ since $\im B \subseteq \im C$.
\begin{Ex} The domain of the rational operator $H_1$  with decomposition $u(\cS-1) (\frac{1}{u}(\cS+1))^{-1} $ introduced for the modified Volterra chain (\ref{mvalham}) is
$$\cF'_{H_1}=\{ f \in \cF' | \sum_n \frac{\partial f_n}{\partial u} \in \frac{1}{u}(S+1) \cF \}.$$
\end{Ex}
It follows from $H^\dagger=-H$ that
\begin{equation}
 \label{AB+}
 B^\dagger A=-A^\dagger B.
\end{equation}

The pair $A,B$ naturally defines an anti-symmetric bracket $\{ \bullet, 
\bullet \}_{A,B}: \cF'_B \times \cF'_B \mapsto \cF'_B$. For $f,g\in\cF'_B$ there 
exist $a,b\in\cF$ such that $Ba=\delta_u f$ and $Bb=\delta_u g$. Then the
bracket $\{ f, 
g \}_{A,B} $ can be defined as follows (c.f.(\ref{lieham}))
\begin{equation}
\{  f,  g\}_{A,B}=\smallint Ba \cdot Ab.
\end{equation}
It is independent on the choice of $a$ and $b$. Indeed, 
\begin{equation*}
\smallint \delta_u f \cdot Ab=\smallint Ba\cdot Ab=\smallint 
a\cdot B^{\dagger}Ab=-\smallint a\cdot A^{\dagger}Bb=-\smallint 
Aa\cdot \delta_u g,
\end{equation*}
 since  $A^{\dagger}B$ is skew-symmetric (\ref{AB+}). This also 
implies that the bracket $\{ \bullet, \bullet \}_{A,B}$ itself is anti-symmetric: 
 \begin{equation*}
 \{  f,  g\}_{A,B}=\smallint \delta_u f \cdot Ab=-\smallint 
Aa\cdot \delta_u g=-\{ 
g, f\}_{A,B}.
 \end{equation*}

\begin{Pro}\label{pronoha2}
Let $A$ and $B$ be two difference operators such that their ratio $AB^{-1}$ is skew-symmetric and such that the bracket $\{ \bullet, \bullet \}_{A,B}$ is a Lie bracket 
on 
$\cF'_B$. Assume that the form $\smallint(r\cdot\delta_u f)$, where 
$r~\in~\cF,\ f \in \cF'_B$ is non-degenerate. Then the operator $A$ is 
preHamiltonian 
satisfying 
\begin{equation}\label{apre1}
A(\omega_A(a,\bullet))=A_*[Aa]-(D_A)_aA,\qquad \forall a\in\cF,
\end{equation}
 and the operator $B$ satisfies  
\begin{equation} \label{eqnlochambis1}
 B_*[A a]-(D_B)_a A+(D_B)_a^{\dagger} A+(D_A)_a^{\dagger} B=B 
(\omega_A(a,\bullet)) ,\qquad \forall a\in\cF.
\end{equation}
\end{Pro}

\begin{proof}
We know that $\{ \bullet, \bullet \}_{A,B}$ is  a Lie bracket on $\cF'_B$, which 
implies that $\{  f,  g\}_{A,B} \in \cF'_B$ for all $ f,  g \in \cF'_B$ and thus 
$\delta_u \{  f,  g\}_{A,B} \in B(\cF)$. 
Let $W$ be the $\cK$--linear space $W=\{a\in\cF\,|\, 
(Ba)_*=(Ba)_*^\dagger\}$ , or in other words for any element $a\in W$ there 
exists $f\in\cF'_B$ such that $Ba=\delta_u f$. The space $W$ is infinite 
dimensional over $\cK$ since the form $\smallint(r\cdot\delta_u f)$ is 
non-degenerate.
For all $a, b \in W$ we 
have
\begin{equation*}
\begin{split}
 &\quad \delta_u(Ba\cdot Ab)
  = (Ba)_*^{\dagger}(Ab)+(Ab)_*^{\dagger}(Ba)
    = (Ba)_*[Ab]+(Ab)_*^{\dagger}(Ba) \\
   &= B_*[Ab](a)+B a_*[A 
b]+b_*^{\dagger}A^{\dagger}B(a)+(D_A)_\fb^{\dagger}(Ba)\\ 
   &= B_*[Ab](a)+B a_*[A 
b]-b_*^{\dagger}B^{\dagger}A(a)+(D_A)_b^{\dagger}(Ba)\\
   &= B_*[Ab](a)+B a_*[A 
b]-(Bb)_*^{\dagger}(Aa)+(D_B)_b^{\dagger}(Aa)+(D_A)_b^{\dagger}
(Ba)\\ 
    &= B_*[Ab](a)+B a_*[A 
b]-(Bb)_*(Aa)+(D_B)_b^{\dagger}(Aa)+(D_A)_b^{\dagger}(Ba)\\ 
    &= \left(B_*[Ab]-(D_B)_b A+(D_B)_b^{\dagger} 
A+(D_A)_b^{\dagger}B\right)(a) +B \left(a_*[A b]-b_*[A a] \right).
    \end{split}
    \end{equation*}
This implies the existence of a form $\omega$ such that for all $a \in \cF$,
 \begin{equation*} 
 B_*[A a]-(D_B)_a A+(D_B)_a^{\dagger} A+(D_A)_a^{\dagger} B=B 
(\omega(a,\bullet)).
\end{equation*}
Indeed, if $M$ is a bidifference operator such that $M(a,b) \in \im B$ for all 
$a, b \in V$, where $V$ is a subspace of $\cF$ infinite-dimensional over $\cK$, 
then there exists a bidifference operator 
$N$ such that $M(a,b)=B(N(a,b))$ for all $a, b \in \cF$. In terms of $\omega$ we 
have for all $a,b \in W$
\begin{equation}\label{xyinW}
B(\omega(a,b)+b_*[Aa]-a_*[Ab])=\delta_u(Bb\cdot Aa).
\end{equation}
 Let $f, g , h \in \cF'_B$ be such that $\delta_u f=Ba$, $\delta_u 
g=Bb$, and $\delta_u h=Bc$ for some $a, b, c \in W$. The first 
term in the Jacobi identity is
\begin{equation*}
\{f,\{g,h\}_{A,B}\}_{A,B} =-\smallint B(a)\cdot A(\omega(b, 
c)+c_*[Ab]-b_*[A\hb])
\end{equation*}
The second term is:
\begin{equation*}
\begin{split}
\{g,\{h, f\}_{A,B}\}_{A,B}& =\smallint Ab\cdot \delta_u(Bc\cdot Aa) =-\smallint 
Ab\cdot \delta_u(Ba\cdot Ac)\\
                  &=-\smallint Ab\cdot (Ba)_*^{\dagger}(Ac)-\smallint 
Ab\cdot (Ac)_*^{\dagger}(Ba) 
                  =-\smallint Ab\cdot (Ba)_*[Ac]-\smallint Ba\cdot 
(Ac)_*[Ab].
                  \end{split}
\end{equation*}
and similarly, the third term is
\begin{equation*}
\begin{split}
\{h,\{f, g\}_{A,B}\}_{A,B} &=\smallint Ac\cdot \delta_u(Ba\cdot Ab) 
                  =\smallint Ac\cdot (Ba)_*[Ab]+\smallint Ba\cdot 
(Ab)_*[Ac].
\end{split}
\end{equation*}
Hence we get 
\begin{equation}\label{jach}
\smallint Ba\cdot \left( A(\omega(b, 
c)+c_*[Ab]-b_*[Ac])+(Ab)_*[Ac]-(Ac)_*[Ab]\right) =0.
\end{equation}
Therefore 
\begin{equation}\label{Apreham}
A\left(\omega(a, b)\right)=A_*[Aa](b)-A_*[Ab](a)
\end{equation}
for all $a, b \in W$. Since $W$ is infinite-dimensional over $\cK$, 
\eqref{Apreham} holds for all $a, b \in \cF$, which is to say that $A$ is 
preHamiltonian. 
\end{proof}

The converse statement is also true and it does not require the minimality of 
the decomposition. 
\begin{Pro}\label{pronoha1}
Let $A$ and $B$ be two difference operators such that their ratio $H=AB^{-1}$ is skew-symmetric. Assume that 
the operator $A$ is preHamiltonian , i.e. there exist a $2$-form $\omega_A$ such that
(\ref{apre1}) holds
and that the operator $B$ satisfy the equation
 (\ref{eqnlochambis1}).
 Then the bracket $\{ \bullet, \bullet \}_{A,B}$ is a Lie bracket on 
$\cF'_B$. 
\end{Pro}
\begin{proof}
The bracket $\{ \bullet, \bullet \}_{A,B}$ is well-defined on $\cF'_B$. Indeed, for 
all $a, b \in W$ equation \eqref{xyinW} holds. Moreover, since $A$ is 
preHamiltonian, equation \eqref{jach} is 
satisfied for all $a, b, c \in W$. Therefore the bracket $\{ \bullet, 
\bullet \}_{A,B}$ satisfies the Jacobi identity.
\end{proof}

Proposition \ref{pronoha1} can be seen as an analogue of Proposition $7.8$ 
in \cite{DSK13} in the case of rational differential Hamiltonian operators, 
which has been proven by methods of Poisson Vertex Algebras. Note that in the 
proof of Proposition \ref{pronoha1} we do not make any 
assumptions on the dimension of the space $\cF'_B$. In particular we do not require the form 
$\smallint(r\cdot\delta_u f)$ to be non-degenerate. Although the properties 
of the Poisson bracket, such as anti-symmetry and Jacobi identity have to be 
verified only on the elements of $\cF'_B$, the operator identities obtained are 
satisfied on all elements of $\cF$. This reflects the {\sl Substitution 
Principle} (see \cite{mr94g:58260}, Exercise 5.42).  In general it is very difficult to characterise the space 
$\cF'_H$, the Substitution Principle enables us to check the identities over 
the difference field $\cF$. Having it in mind and as well the Propositions 
\ref{pronoha2}, \ref{pronoha1} we can give a new and easily verifiable 
definition of a rational Hamiltonian operator.

\begin{Def}\label{nonlocham}
Let $H$ be a skew-symmetric rational operator. We say that $H$ is Hamiltonian if there exists a decomposition $H=AB^{-1}$ such that the operator $A$ is 
preHamiltonian, i.e. if there is a $2$-form $\omega_A$ such that for all $a \in 
\cF$
\begin{equation}\label{apre}
A(\omega_A(a,\bullet))=A_*[Aa]-(D_A)_aA
\end{equation}
 and if the operators $A$ and $B$ satisfy
 \begin{equation} \label{eqnlochambis}
 B_*[A a]-(D_B)_a A+(D_B)_a^{\dagger} A+(D_A)_a^{\dagger} B=B 
(\omega_A(a,\bullet)) \text{  for all } a \in \cF.
\end{equation}
\end{Def}
\begin{Rem}
Note that if a decomposition $H=AB^{-1}$ satisfies equations \eqref{apre} and \eqref{eqnlochambis}, then so does a minimal decomposition of $H=A_0{B_0}^{-1}$. 
Indeed if a pair of difference operator $A$, $B$ such that $A$ is prehamiltonian and equation \eqref{eqnlochambis} is satisfied has a common right factor, i.e. $A=A_0C$ and $B=B_0C$,  
then $A_0$ is preHamiltonian and the pair $A_0, B_0$ satisfies \eqref{eqnlochambis} as well. 
\end{Rem}
\begin{Rem}
Taking $B=1$ in Definition \ref{nonlocham} of rational Hamiltonian operators, one recovers the Definition \ref{locham} of Hamiltonian difference operator. In the sequel, we will say Hamiltonian operator to refer to a (a priori rational) operator in $\cQ$ satisfying Definition \ref{nonlocham}.
\end{Rem}
Definition \ref{nonlocham} can also be viewed as direct generalisation of Theorem \ref{algham} as explained in the following statement.
\begin{Pro}\label{pnlham}
Let $H$ be a skew-symmetric rational operator with minimal decomposition $H=AB^{-1}$. If $H$ satisfies (\ref{lochbis}) for all $a$ in the images of operator $B$,
then there is a $2$-form $\omega_A$ satisfying (\ref{apre}) and (\ref{eqnlochambis}) for all $a\in \cF$.
\end{Pro}
\begin{proof}
For $H=AB^{-1}$, we have $H_*=A_*B^{-1}-AB^{-1}B_* B^{-1}$. Taking $a=B  b ,  b  
\in \cF$, we get
$$
(D_H)_a=(D_A)_ b -AB^{-1}(D_B)_ b .
$$
Thus identity (\ref{lochbis}) leads to
\begin{eqnarray*}
 A_*[A b ]-A B^{-1} B_*[A b ]-(D_A)_ b  A+AB^{-1}(D_B)_ b  A
 =AB^{-1}\left({(D_A)_ b }^{\dagger} B +{(D_B)_ b }^{\dagger} A  \right),
\end{eqnarray*}
where we used $H$ being anti-symmetric, that is,
\begin{eqnarray}\label{mix}
 A_*[A b ] -(D_A)_ b  A
 =AB^{-1}\left(B_*[A b ]-(D_B)_ b  A+{(D_A)_ b }^{\dagger} B +{(D_B)_ b 
}^{\dagger} A  \right) .
\end{eqnarray}
Let $CA=DB$ be the left least common multiple of the pair $A$ and $B$. It is 
also the right least common multiple of 
the pair $C$ and $D$ since $AB^{-1}$ is minimal. 
It follows from (\ref{mix}) that 
$$
C A_*[A b ] -(D_A)_ b  A
 =D\left(B_*[A b ]-(D_B)_ b  A+{(D_A)_ b }^{\dagger} B +{(D_B)_ b }^{\dagger} A  
\right) 
$$
Therefore there exists a $2$-form denoted by $\omega_A$ satisfying (\ref{apre}) and (\ref{eqnlochambis}).
\end{proof}

\begin{Ex}\label{exam4}
We check that the operator $H_1$ defined by (\ref{mvalham}) is indeed Hamiltonian.
Note that 
$$H_1=A B^{-1},\ A=u(\cS-1), \ B=\frac{1}{u}(\cS+1).$$
It is obviously skew-symmetric.
 For any $a,b \in \cF$ we have $A_*[Aa](b)=u(a_1-a)(b_1-b)$. Hence $A$ is preHamiltonian with $\omega_A=0$. 
 We have $(D_A)_a=a_1-a$ and $(D_B)_a=-\frac{1}{u^2}(a_1+a)$.  Thus we have
\begin{eqnarray*} \label{exmodvol}
B_*[Aa]=-\frac{1}{u}(a_1-a)(S+1), && (D_B)_a A=\frac{1}{u}(a_1+a)(S-1), \\
(D_B)_a^{\dagger} A=\frac{1}{u}(a_1+a)(S-1), &&
(D_A)_a^{\dagger} B=\frac{1}{u}(a_1-a)(S+1).
\end{eqnarray*}
Therefore, (\ref{eqnlochambis}) is satisfied and $H_1$ is a Hamiltonian operator. 
\end{Ex}

We now investigate how preHamiltonian pairs relate to 
Hamiltonian pairs. 

\begin{Pro}\label{proham}
Let $A$ and $B$ be compatible preHamiltonian operators. Assume that there exists a difference 
operator $C$ such that $BC^{-1}$ is skew-symmetric, $B$ and $C$ satisfy \eqref{eqnlochambis} and
$AC^{-1}$ is skew-symmetric. Then the operators $A$ and  $C$ satisfy \eqref{eqnlochambis}.
In particular, the rational operator $AC^{-1}$ is Hamiltonian.
\end{Pro}
\begin{proof} 
Since the difference operators $A$ and $B$ form a preHamiltonian pair, for all $a \in 
\cF$ we have
\begin{eqnarray}
&&A_*[Aa]-(D_A)_a A=AM_a;\label{eqvA} \\
&&B_*[Ba]-(D_B)_a B=BN_a; \label{eqvB}\\
&&A_*[Ba]+B_*[Aa]-(D_B)_a A-(D_A)_aB=AN_a+BM_a, \label{eqvAB}
\end{eqnarray}
where $M_a=\omega_A(a, \bullet)$ and $N_a=\omega_B(a, \bullet)$.
From the assumption, we know that
\begin{eqnarray} 
&& C_*[Ba]+(D_B)_a^\dagger C+(D_C)_a^\dagger B-(D_C)_a B=C N_a. \label{eqvC}
\end{eqnarray}
We need to prove that the operators $A$ and  $C$ satisfy \eqref{eqnlochambis}, that is, 
for all $a \in \cF$,
\begin{equation} \label{toprove}
 C_*[Aa]+(D_A)_a^\dagger C+(D_C)_a^\dagger A-(D_C)_a A= CM_a .
\end{equation}
Let $\Sigma$ be the difference of the LHS with the RHS of
\eqref{toprove}. We are going to show that both $A^\dagger \Sigma$ and
$B^\dagger \Sigma$ are skew-symmetric. We know that the rational operators $AC^{-1}$ and $BC^{-1}$ are skew-symmetric,
that is, $A^\dagger C$ and $B^\dagger C$ are skew-symmetric.
We first prove that $A^\dagger\Sigma$ is skew-symmetric. In the following we 
use the notation $ \equiv$ to say that two operators are equal up to adding an 
skew-symmetric operator. We have
\begin{equation*} 
\begin{split}
A^\dagger\Sigma &=A^\dagger C_*[Aa]+A^\dagger(D_{A})_a^\dagger 
C+A^\dagger(D_{C})_a^\dagger A-A^\dagger(D_{C})_aA -A^\dagger C M_a \\
                       & \equiv A^\dagger C_*[Aa]+A^\dagger(D_{A})_a^\dagger C 
-A^\dagger CM_a  \\
                       & \equiv -A_*^\dagger[Aa] C+A^\dagger (D_{A})_a^\dagger 
C+C^\dagger AM_a \\
                       & \equiv -M_a^\dagger A^\dagger C+C^\dagger AM_a 
                        \equiv 0 .
                       \end{split}
\end{equation*}
since $A^\dagger C$ is a skew-symmetric operator and $A$ is a preHamiltonian operator. We now check that 
$B^\dagger \Sigma$ is also skew-symmetric :
\begin{equation*}
\begin{split}
B^\dagger \Sigma &=B^\dagger C_*[Aa]+B^\dagger (D_{A})_a^\dagger 
C+B^\dagger(D_{C})_a^\dagger A-B^\dagger(D_{C})_aA -B^\dagger C M_a\\
& \equiv -B_*^\dagger[Aa] C+B^\dagger (D_{A})_a^\dagger 
C+B^\dagger(D_{C})_a^\dagger A-B^\dagger(D_{C})_aA -B^\dagger CM_a \\
& \equiv A_*[Ba]^\dagger C-A^\dagger(D_{B})_a^\dagger C-N_a^\dagger A^\dagger 
C-M_a^\dagger B^\dagger C \\
&\quad +B^\dagger(D_{C})_a^\dagger A-B^\dagger(D_{C})_aA +C^\dagger BM_a\\
& \equiv -A^\dagger C_*[Ba]-A^\dagger (D_{B})_a^\dagger C+N_a^\dagger C^\dagger 
A+B^\dagger(D_{C})_a^\dagger A-B^\dagger(D_{C})_aA  \\
& \equiv -A^\dagger C_*[Ba]-A^\dagger(D_{B})_a^\dagger C+C_*[Ba]^\dagger 
A+C^\dagger(D_{B})_aA \\
& \equiv 0.
\end{split}
\end{equation*}
We used relations (\ref{eqvB})--(\ref{eqvC}) and the fact that $A^\dagger C$ and $B^\dagger C$ are skew-symmetric 
operators. By now, we have proved that
\begin{equation}\label{sigma}
A^\dagger \Sigma =-\Sigma^\dagger A \quad \mbox{and} \quad
B^\dagger \Sigma =-\Sigma^ \dagger B.
\end{equation}
This leads to that for all 
$\lambda \in \cK$ we have $(A+\lambda B)^\dagger \Sigma=-(\Sigma)^\dagger(A+\lambda B)$. 
By Lemma \ref{CD} it implies that $(A+\lambda B)$ 
divides $\Sigma$ on the right for all $\lambda \in \cK$. If $\Sigma \neq 0$, it follows from Lemma \ref{CDE} that there exists $b \in \cF$ and $X \in \cR$ such that $XB=bXA$. 
Hence we have $H=AC^{-1}$ and $X^{-1}bX H=BC^{-1}$ are both skew-symmetric, that is $b \tilde{H}=\tilde{H} b$ where $\tilde{H}=XHX^{\dagger}$. 
This can only be the case if $b\in\cK$ is a constant. 
But in this case we have nothing to prove.  Thus $\Sigma=0$ implying that $AC^{-1}$ is 
a Hamiltonian operator by Definition \ref{nonlocham}.
\end{proof}

The above proposition shows that for a preHamiltonian pair $A$ and $B$, if there is a difference operator $C$ such that
the ratio with one of them is a Hamiltonian operator, so is the ratio with another one if it is skew-symmetric. 
We will give much stronger result in the following theorem:

\begin{The}\label{thham}
Let $A$ and $B$ be compatible preHamiltonian operators and $K$ be a 
rational Hamiltonian operator. Then, provided that
$H=AB^{-1}K$ is skew-symmetric, it is Hamiltonian and compatible with $K$.
\end{The}

\begin{proof}
Let $CD^{-1}$ be a minimal decomposition of $K$.
We start by writing $B^{-1}C$ as a right fraction using the Ore condition 
$BG=CK$. We only need to check that $AG$ and $BG$ are compatible preHamiltonian 
operators and that the pair $CK$ and $DK$ satisfies \eqref{eqnlochambis}. Since 
$H=(AG)(DK)^{-1}$, we will then be able to conclude using Proposition \ref{proham}. 

We are going to prove that $AG$ is preHamiltonian by making use of Lemma 
 \ref{rlcmpreh}: if two difference operators are preHamiltonians, then 
their rlcm is preHamiltonian as well. The key is to write 
$AG$ as the rlcm of two preHamiltonian operators. A priori $B$ and $C$ do not need to
be left coprime. Let us write $B=E \tilde{B}$ and $C=E \tilde{C}$, where 
$\tilde{B}$ and $\tilde{C}$ are left coprime. Since $H$ and $CD^{-1}$ are skew-symmetric, we have $K^\dagger D^\dagger AG=-G^\dagger A^\dagger DK$ and 
$C^\dagger D=-D^\dagger C$.
Therefore $D^\dagger AG=\tilde{C}^\dagger X$ and $A^\dagger DK=-\tilde{B}^\dagger X$ for some difference 
operator $X$ since we have $G^\dagger B^\dagger=K^\dagger C^\dagger$. $C$ and $D$ are right coprime, 
hence a fortiori $\tilde{C}$ and $D$ are right coprime. It follows that $D^\dagger $ and $\tilde{C}^\dagger $ are left coprime. Therefore there exist two 
difference operators $Y$ and $Z$ with $\Ord \,Y=\Ord \, \tilde{C}$ and $\Ord \, Z=\Ord \, D$ such 
that $D^\dagger Y=\tilde{C}^\dagger  Z$ is the rlcm of $D^\dagger$ and $\tilde{C}^\dagger$. From 
$C^\dagger D=-D^\dagger C$ we see that $Y$ divides $C$ on the left hence that it is 
preHamiltonian (any left factor of a preHamiltonian operator is preHamiltonian). 
$AG$ is the rlcm of $Y$ and $A$. Indeed, from $D^\dagger AG=\tilde{C}^\dagger X$ we see that 
$Y$ divides $AG$ on the left. Moreover $\Ord \,Y=\Ord \, \tilde{C}=\Ord \, G$ by definition of 
$G$, $Y$ and $\tilde{C}$. $AG$ is the rlcm of two preHamiltonian operators, 
hence it is preHamiltonian. 

The exact same argument to get $AG$ being preHamiltonian can be applied to $H+ 
\lambda CD^{-1}$ for any $\lambda \in \cK$. It amounts to replace $AG$ by $AG+ 
\lambda BG$. Therefore, we have proved that the two difference operators $AG$ 
and $BG$ form a compatible pair of preHamiltonian operators. Let us call $N$ the 
bidifference operator associated to $BG=CK$ (that is to say $\omega_{BG}(a,\bullet)=N_a$ 
for all $a \in \cF$).

Next we want to check that operators $CK$ and $DK$ satisfies \eqref{eqnlochambis}. We 
already know that $CK=BG$ is preHamiltonian, with bidifference operator $N$. 
Hence, we need to verify that for all $a \in \cF$
\begin{equation} \label{eqaux1}
(DK)_*[CKa]+(D_{CK})_a^\dagger DK+(D_{DK})_a^\dagger CK-(D_{DK})_a CK=DKN_a ,
\end{equation}
which follows from 
\begin{equation} \label{eqaux2}
C_{Ka}K=KN_a+(D_K)_aCK-K_*[CKa],
\end{equation}
where $C_a$ is the bidifference operator associated to the preHamiltonian $C$ 
(i.e. $C_a=\omega_C(a, \bullet)$ for all $a \in \cF$). Indeed (recall that $BG=CK$), we have
\begin{equation*}
\begin{split}
CKN_a&=(CK)_*[CKa]-(D_{CK})_aCK \\
           &=C_*(CKa)K-(D_{C})_{Ca}CK+CK_*(CKa)-C(D_{K})_aCK \\
           &=C(C_{Ka}K+K_*(CKa)-(D_{K})_aCK)
           \end{split}
           \end{equation*}
and we can simplify on the left by $C$ since $\cR$ is a domain. One deduces 
\eqref{eqaux1} from \eqref{eqaux2} multiplying on the left \eqref{eqaux2} by $D$ 
and using the fact that the operators $C$ and $D$ satisfy \eqref{eqnlochambis}.
\\
By Proposition \ref{proham}, we obtain that $AG$ and  $DK$ satisfy \eqref{eqnlochambis}. In other words, 
$H=(AG)(DK)^{-1}$ is a Hamiltonian operator under the assumption that it is skew-symmetric.
The same proof holds when replacing $A$ by $A+\lambda B$ for $\lambda \in \cK$ and thus $H$ and $K$ are compatible.
\end{proof}
This result is very strong.  Theorem \ref{repair} corresponds to the special case when the 
 Hamiltonian operator $K =BC$.

\begin{Ex}
Consider the Narita-Itoh-Bogayavlensky
lattice \cite{narita,itoh,Bog88} of the form
\\
$u_t=u(u_1 u_2-u_{-1} u _{-2}).$
It possesses a Nijenhuis recursion operator \cite{wang09}
\begin{eqnarray*}
&&R=u(\cS^2\!-\!1)^{-1}(\cS\!-\!\cS^{-2})(\cS^2 uu_1-u u_1 \cS^{-1}) (\cS u u_1-u u_1 \cS^{-1})^{-1} (u_1 u_2\cS^3-u u_1 )  (\cS-1)^{-1}\frac{1}{u}\\
&&\quad=AB^{-1}, \qquad B=u(\cS-1)\sum_{n=0}^3 B^{(n)}\cS^n,
\end{eqnarray*}
where
\begin{eqnarray*}
&& B^{(0)}=-u_{-3} u_{-1} u_{-4}^2+u_{-2} u_{-1}^2 u_{-4}-u_{-5} u_{-3} u_{-2} 
u_{-4}+u_{-5} u_{-2} u_{-1} u_{-4}
\\&&\qquad \quad-u_{-3} u_{-1} u u_{-4}+u_{-3} u_{-2} u_{-1} u\\
&&B^{(1)}=u_{-1} u^2 u_{-3}+u_{-4} u_{-1} u u_{-3}+u_{-2} u_{-1} 
u u_{-3}-u_{-2}^2 u_{1} u_{-3}-u_{-4} u_{-2} u_{1} u_{-3}-2 u_{-2} u
u_{1} u_{-3}\\&&\qquad\quad-u_{-2} u_{1} u_{2} u_{-3}-u_{-3}^2 u_{-2} u+u_{-2}^2 
u_{-1} u_{1}+u_{-2} u_{-1} u u_{1}+u_{-2} u_{-1} u_{1} u_{2} \\ 
&&B^{(2)}=u u_{2} 
u_{-1}^2-u_{1} u_{2}^2 u_{-1}+u_{-2} u u_{1} u_{-1}-2 u_{-2} u_{1} u_{2} 
u_{-1}+u u_{1} u_{2} u_{-1}+u u_{2} u_{3} u_{-1}+u_{-2} u 
u_{1}^2\\&&\qquad\quad+u_{-3} u_{-2} u u_{1}-u_{-2} u_{1}^2 u_{2}-u_{-3} u_{-2} u_{1} 
u_{2}-u_{-2} u_{1} u_{2} u_{3} \\ 
&&B^{(3)}=u_{1} u_{3} u^2-u_{2} u_{3}^2 
u+u_{-1} u_{1} u_{2} u-u_{-1} u_{2} u_{3} u+u_{1} u_{3} u_{4} 
u-u_{1} u_{2} u_{3} u_{4}
\end{eqnarray*}
and a rational Hamiltonian operator
$$
H=u \cS^{-1}(\cS^3-1)(\cS+1)^{-1} u,
$$
which can be proved as in Example \ref{exam4} following Definition 
\ref{nonlocham}. Using the procedure described in the proof of Proposition \ref{pro5} in the next section
we show that the operator $B$ is preHamiltonian.
Since $R$ is Nijenhuis, thus $A$ and $B$ form a preHamiltonian pair 
\cite{CMW18a}.
It is easy to verify that $RH$ is 
skew-symmetric, hence by Theorem \ref{thham} the rational operator $RH$ is a Hamiltonian operator.
\end{Ex}

\begin{The}\label{hamp}
Let $H$ and $K$ be two compatible rational Hamiltonian operators. Then there exist two compatible preHamiltonian  operators $A$ and $B$ 
such that $HK^{-1}=AB^{-1}$. 
\end{The}
\begin{proof}
Let $CD^{-1}$ (resp. $PQ^{-1}$) be a minimal presentation of $H$ (resp. $K$). 
Let $DM=QN$ be the least right 
common multiple of $D$ and $Q$ and $ \lambda \in \cK$. Then $H+ \lambda K$ 
which by hypothesis is Hamiltonian can be rewritten as $(CM+\lambda 
PN)(DM)^{-1}$. For infinitely many $\lambda$, $CM+\lambda PN$ and $DM=QN$ are 
right coprime. Hence $CM+\lambda PN$ is preHamiltonian for infinitely many 
constants $\lambda \in \cK$. Notice that  $HK^{-1}=(CM)(PN)^{-1}$. We conclude the proof letting $A=CM$ and $B=PN$.
\end{proof}

Combining Theorem \ref{thham} and \ref{hamp}, we are able to prove the following known statement:
\begin{Cor}
Let $H$ and $K$ be two rational compatible Hamiltonians. Define $L=HK^{-1}$. 
Then $L^nH$ is Hamiltonian for all $n \in \bbbz$.
\end{Cor}

\section{An application to Hamiltonian integrable Equations}\label{Sec5}
In our recent paper \cite{CMW18a} we constructed a recursion operator for the 
Adler-Postnikov equation
\cite{Adler2}
\begin{equation}\label{bv}
 u_t=u^2(u_2 u_1-u_{-1} u_{-2})-u (u_1-u_{-1}):=c
\end{equation}
using its (rational) Lax
representation. In this section, we show that it is a Hamiltonian system. We start by introducing some 
relevant basic definitions for differential-difference equations.

Thus there is a bijection between evolutionary equations
\begin{equation}\label{evol}
 u_t=a,\qquad a\in\cF
\end{equation}
and evolutionary derivations of $\cF$. With equation (\ref{evol}) we associate 
the vector field $X_a$.

\begin{Def} An evolutionary vector field with characteristic $b \in\cF$ is a symmetry of  the system
(\ref{evol}) if and only if 
 $[b,\ a]=0.$
\end{Def}

The space of symmetries of an equation forms a Lie algebra. The existence 
of an infinite dimensional commutative Lie algebra of symmetries is a 
characteristic property of an integrable equation and it can be taken as a 
definition of integrability.

Often the symmetries of integrable equations can be generated by recursion 
operators
\cite{mr58:25341}. Roughly speaking, a recursion operator
is a linear operator $R: \cF\rightarrow \cF$ mapping a symmetry to a new 
symmetry.  For an evolutionary equation 
(\ref{evol}), it satisfies the following equation in $\cQ$
\begin{equation}\label{reopev}
R_t=R_*[a]=[ a_*,\ R] .
\end{equation}
It was shown in \cite{Carp2017} for the differential  case and in \cite{CMW18a} 
for the difference case that a necessary condition for a rational operator $R$ 
to generate an infinite dimensional commutative Lie algebra of symmetries is to 
admit a decomposition $R=AB^{-1}$ with $A$ and $B$ compatible preHamiltonian. It 
follows then that $R$ is Nijenhuis, and in particular that $R$  is also a
recursion operator for each of the evolutionary equations in the hierarchy 
$u_t=R^k(a)$, where $k=0,1,2,\ldots \ .$ An alternative method for proving the 
locality of the hierarchy generated by a Nijenhuis operator is given in 
\cite{wang09}.
\begin{Def}
An evolutionary equation (\ref{evol}) is said to be a hamiltonian equation if 
there exists a Hamiltonian operator $H$
 and a hamiltonian functional $\int f \in \cF'$ such that
$u_t=a=H \delta_{u} \int f. $
\end{Def}
This is the same as to say that the evolutionary vector field $a$ is a 
hamiltonian vector field and thus the Hamiltonian operator
is invariant along it, that is, 
\begin{equation}\label{hameq}
 H_t=H_*[a]=a_* H+H a_*^{\dagger} ,
\end{equation}
which follows immediately from equation (\ref{lochbis}) and the fact that for $b \in \cF$, 
$b_*={b_*}^{\dagger}$ if and only if $b$ is a variational derivative.

We now recall some relevant results for the equation (\ref{bv}) in \cite{CMW18a}.  The equation (\ref{bv}) possesses a recursion operator:
\begin{eqnarray}
&&R=u\left(u (\cS^2\!-\!\cS^{-1}\!) u +\cS^{-1}\!\!\!-\!1\right)(\cS 
u\!-\!u\cS^{-1}\!)^{-1}\left( u (\cS\!-\!\cS^{-2}\!)u 
(\cS^2+\cS+1)+\cS^2\!-\!\cS \right)(\cS^2\!-\!1)^{-1}\! \frac{1}{u}\nonumber\\
&&\quad +u(2 \cS^{-1} u-\cS^{-2} u-\cS u+u-u\cS) 
(\cS^2+\cS+1)(\cS^2-1)^{-1}\frac{1}{u}. \label{readler1}
\end{eqnarray}
The rational operator $R$ can be factorised as 
$R=AB^{-1}$ where the operators $A$ and $B$ form a preHamiltonian pair. We have proved the following statement:

\begin{quote}
There exists $d^{(n)} \in \cF$, $n \geq 1$ such that  $c^{(n+1)}=B(d^{(n+1)})=A(d^{(n)}) \in \mathrm{K}$ for all $n$ and 
$[c^{(n)},\ c^{(m)}]=0$ for all $n, m \geq 1$.
\end{quote}
In particular, for the equation (\ref{bv}), we have $c=c^{(1)}=B(d^{(1)})$, where $d^{(1)}=\frac{u_{-1}uu_1w_{-1}ww_1}{\alpha \gamma}$
and $c^{(n)}=R^{n-1}c$.

In what follows, we show that the system \eqref{bv} is hamiltonian.
Let $H$ be the following skew-symmetric rational operator
\begin{equation} \label{ham}
\begin{split}
H&=u^2u_1u_2^2\cS^2-\cS^{-2}u^2u_1u_2^2+\cS^{-1}uu_1(u+u_1)-uu_1(u+u_1)\cS \\
&+u(1-\cS^{-1})(1-uu_1)({\cS}u-u\cS^{-1})^{-1}(1-uu_1)(\cS-1)u.
\end{split}
\end{equation}
Note that the equation (\ref{bv}) can be written in the form $u_t=H(\frac{\delta 
(\ln u)}{\delta u})$. We are going to prove that $H$ is a Hamiltonian operator.

The operator (\ref{ham}) can be represented in the factorised form 
$H=CG^{-1}$, where $C,G$ are difference operators. Indeed, it is easy to verify 
that 
\[
 (1-uu_1)(\cS-1)u G=({\cS}u-u\cS^{-1})E 
\]
where 
\begin{eqnarray*}
&& G=u_{1} v_{2} (u_{2} v_{1} - u_{1} v_{3}) \cS - (u^2 v_{2} v_{-1} - 
u_{1} 
u_{-1} v_{1} v)  + 
 u_{-1} v_{-1} (u_{-2} v - u_{-1} v_{-2}) \cS^{-1},\\
&& E=v_{2} v u (u_{-1} v_{1} - u v_{-1})  + v_{2} v u_{1}(u_{2} 
v_{1} - u_{1} v_{3})\cS
\end{eqnarray*}
and $v=1-u_{-1}u$. Thus $H=CG^{-1}$, where
\[
C=(u^2u_1u_2^2\cS^2-\cS^{-2}u^2u_1u_2^2+\cS^{-1}
uu_1(u+u_1)-uu_1(u+u_1)\cS)G+u(1-\cS^{-1})(1-uu_1)E.
\]
We have $C=  C^{(3)}\cS^3+\cdots + C^{(-3)}\cS^{-3}$ where
\[
 C^{(3)}=u ^2 u_1  u_2 ^2 u_3  (1 - u_3  u_4 ) ( u_4 -u_3   - u_2  u_3  u_4  
+ u_3  u_4  u_5 ).
\]

\begin{Pro}\label{pro5}
 The operator $H$ given by (\ref{ham}) is a Hamiltonian operator.
\end{Pro}
\begin{proof}
We prove the statement by a direct computation. First we need to show that 
$C$ is a preHamiltonian difference operator. Namely, we need to prove the 
existence of the form 
$\omega_C(a,b)=\sum_{n>m}\omega_{n,m}(\cS^n(a)\cS^m(b)-\cS^n(b)\cS^m(a)),\ 
\omega_{n,m}\in \cF$ satisfying the equation
\begin{equation}
 \label{Cpre}
 C(\omega_C(a,b))=C_*[C(a)](b)-C_*[C(b)](a),\qquad \forall a,b\in\cF
\end{equation}
and find its entries $\omega_{n,m}$ explicitly. The order of the operator 
$C$ is $(-3,3)$. The right hand side of (\ref{Cpre}) is a difference operator of 
order $(-8,8)$ acting on $a$ (same for $b$). Thus $\omega(a,b)$ should be a 
difference operator of order $(-5,5)$ acting on $a$ (same for $b$). The equation 
(\ref{Cpre}) represents the over-determined system of $50$ linear difference 
equations on $20$ non-zero entries $\omega_{n,m}$. We  order this system 
of equations according to the lexicographic ordering for products of variables 
$a_i=\cS^i (a),b_j=\cS^j(b)$, namely $a_i b_j>a_n b_m$ if $i>n$ or  if $i=n$ and 
$j>m$. In this ordering of the basis the equations on $\omega_{n,m}$ have a 
triangular form and can be solved consequently. The highest equation, 
corresponding to $a_8b_3$, is of the form
\[
  C^{(3)}\cS^3(\omega_{5,0})=u^2 u_1  u_2 ^2 u_3 ^2 u_4  (-1 + u_3  u_4 ) 
u_5 ^2 u_6  u_7 ^2 u_
  8  (-1 + u_8  u_9 ) (u_8  - u_9  + u_7  u_8  u_9  - u_8  u_9  u_{10} ).
\]
Thus 
\[
 \omega_{5,0}=\frac{u u_1 u_2^2 u_3 u_4^2 u_5 (1-u_5u_6) (u_4 u_5 
u_6-u_5 u_6 u_7 +u_5-u_6)}{u_{-1} u_1 u-u_1 u_2 u+u-u_1}.
\]
Consequently we can find all twenty nonzero entries 
$\omega_{n,m}=-\omega_{m,n}$ where $$n>m, \qquad \ -5\le m\le 0\le n\le 5,\qquad 
\ 5\le n-m\le n$$
and check the consistency of the system (\ref{Cpre}). In order to complete the 
proof we need to verify the identity
\[
 G_*[C a]-(D_G)_a C+(D_G)_a^{\dagger} C+(D_C)_a^{\dagger} G=G 
(\omega_C(a,\bullet)) \text{  for all } a \in \cF,
\]
which we have done done by a direct substitution.
\end{proof}

\begin{The}
Let $K=RH$. Then $K$ is skew-symmetric and hence $K$ is a Hamiltonian operator.
\end{The}
\begin{proof}
$R$ is a recursion operator for $u_t=c$, which means that 
\begin{equation} \label{rec}
R_*[c]=c_*R-Rc_*.
\end{equation}
$H$ is a Hamiltonian operator and $u_t=c$ is Hamiltonian for $H$ with density $\ln u$, which means that
\begin{equation*} 
H_*[c]=c_*H+H{c_*}^{\dagger}.
\end{equation*}
It is immediate that 
\begin{equation*}
K_*[c]=c_*K+K{c_*}^{\dagger}.
\end{equation*}
Let $L=K+K^{\dagger}$. We want to check that $L=0$. We have
\begin{equation*} 
L_*[c]=c_*L+L{c_*}^{\dagger}.
\end{equation*}
If we consider the degree of $u$ and its shifts, we can write $K=K^{(0)}+K^{(2)}+K^{(4)}+K^{(6)}+K^{(8)}$. Moreover,
$K^{(0)}=R^{(-1)}H^{(1)}$ and $K^{(8)}=R^{(3)}H^{(5)}$ are obviously skew-symmetric since $R^{(-1)}$ (resp. $R^{(3)}$) is recursion for $u_t=u(u_1-u_{-1})$ (resp. $u_t=u^2(u_1u_2-u_{-1}u_{-2})$) and $u_t=u(u_1-u_{-1})$ (resp. $u_t=u^2(u_1u_2-u_{-1}u_{-2})$) is hamiltonian for $H^{(1)}$ (resp. $H^{(5)}$). Therefore we can write $L=L^{(2)}+L^{(4)}+L^{(6)}$. 
Let $a=u^2(u_1u_2-u_{-1}u_{-2})$. Then
\begin{equation*} 
{L^{(6)}}_*[a]=a_*L^{(6)}+L^{(6)}{a_*}^{\dagger}.
\end{equation*}
If $P$ is a Laurent series in $S^{-1}$ such that its coefficients are homogeneous of degree $n$ and
\begin{equation*} 
P_*[a]=a_*P+P{a_*}^{\dagger}.
\end{equation*}
it is straightforward looking at the leading term of $P$ to see that $n=3k+2$ for some integer $k$. In that case, the order of $P$ is $2k$. Therefore $L^{(6)}=L^{(4)}=0$ and $L=L^{(2)}$, with $L$ of order $0$. But we must also have 
\begin{equation*} 
L_*[b]=b_*L+L{b_*}^{\dagger}.
\end{equation*}
where $b=u(u_1-u_{-1})$. But the only solution to this equation is $L=u(S-S^{-1})u$ which has order $1$. Therefore $L=0$.
We know that $R=AB^{-1}$ and $A, B$ form a preHamiltonian pair. It follows from Theorem \ref{thham} that the operator $K=RH$ is Hamiltonian.
\end{proof}

\begin{The}
Let $\phi \in \cK[X,X^{-1}]$. Then $\phi(R)H$ is a Hamiltonian operator.
\end{The}
\begin{proof}
We know from \cite{CMW18a} that $R$ generates the hierarchy of  \eqref{bv}. therefore $\phi(R)$ also generates an arbitrary large set of commuting flows. From Theorem $5$ in \cite{CMW18a} we see that a minimal decomposition of $\phi(R)$ must come from a pair of compatible preHamiltonian operators. Moreover, $\phi(R)H$ is skew-symmetric. We conclude with Theorem $4$.
\end{proof}

\begin{The}
Every element in the hierarchy $u_{t_n}=c^{(n)}$ is a Hamiltonian system with respect to the Hamiltonian operator $H$.
\end{The}

\begin{proof}
We proceed by induction on $n \geq 1$. We know that the first equation can be written as $u_t=H(1/u)$. Let us assume that the first $n$ equation are hamiltonian for $H$. Let $u_t=a$ be the $n$-th equation of the hierarchy and $u_t=b$ be the $(n+1)$-th equation. By the induction hypothesis, $a=H(f)$ for some variational derivative $f$. Moreover, $b=R(f)$. Since $RH=HR^{\dagger}$ and the total order of a reduced denominator for $RH$ is $6$, $b$ is in the image of a reduced numerator for $H$. As was noticed at the end of section $4$, since $b$ is in the image of a reduced numerator of $H$, it is equivalent to say that $u_t=b$ is hamiltonian for $H$ and to say that 
\begin{equation}
H_*[b]=b_* H+H{b_*}^{\dagger}.
\end{equation}
Since $R^{n-1}$ is recursion for $u_t=b$, the previous equation is equivalent to
\begin{equation}
(R^{n-1}H)_*[b]=b_* R^{n-1}H+R^{n-1}H{b_*}^{\dagger}.
\end{equation}
But this holds thanks to the same principle: $R^{n-1}H$ is a Hamiltonian rational operator and $u_t=g=R^{n-1}c$ is hamiltonian for $R^{n-1}H$. 
\end{proof}

\begin{Rem}
For instance, the second equation of the hierarchy is 
\begin{equation}
u_t=H(u_1u_2+u_{-1}u_1+u_{-1}u_{-2}-1)=H \delta_u (uu_1u_2-u).
\end{equation}
 Moreover, it is easy to check that $R^{\dagger}(1/u)=u_1u_2+u_{-1}u_1+u_{-1}u_{-2}-1$.
\end{Rem}

\section{Concluding remarks}
In this paper, we have developed the theory of Poisson brackets, Hamiltonian 
 rational  operators and difference
preHamiltonian operators associated with a difference field  $(\cK,\{u\},\cS)$, 
where $\cK$ is a zero characteristic ground field of constants, $\{u\}=\{\ldots 
u_{-1},u=u_0,u_1,\ldots\}$ is a sequence of a single ``dependent'' variable and 
$\cS$ is the sift  automorphism of infinite order such that $\cS: \, 
u_k\mapsto u_{k+1}$. This formalism is suitable for the description of scalar 
Hamiltonian dynamical systems.

  It can be extended to the case of several 
dependent variables, i.e. the case when ${\bf u}=(u^1,\ldots ,u^N)$ is a 
vector. 
Some of the definitions concerning the algebra of difference and rational 
operators in the vector case were presented in \cite{CMW18a}. The majority of 
the definitions and results of the current paper can be extended to the vector 
case.  This includes the crucial Propositions  \ref{pronoha2} and \ref{pronoha1}, Theorems 
\ref{thham} and \ref{hamp} as well as Definitions $1$ to $8$. Rational matrix 
difference operators consist of ratios $AB^{-1}$ where $B$ is a regular matrix 
difference operator, that is not a zero-divisor. The fact that the ring of matrix difference operators
is not a domain leads to technical difficulties when trying to extend some of the scalar case results.
In particular, a generalisation of Theorem \ref{charham} to the matrix 
case is not straightforward. 

We want to stress that not all integrable systems of differential-difference 
equations are bi-Hamiltonian. Some systems do possess an infinite hierarchy of 
commuting symmetries generated by a recursion operator which is a ratio of 
compatible preHamiltonian operators, but cannot be cast in a Hamiltonian form 
for any Hamiltonian operator. For example, let us consider the 
equation
\begin{equation}\label{ut1}
 u_{t_1}=u(u_1-u):=f^{(1)}.
\end{equation}
It can be linearised to $v_{t_1}=v_1$ using the substitution 
$u=v_1/v:=\phi$, from which we find its hierarchy of  commuting symmetries, 
corresponding to $v_{t_n}=v_n,\ n\in\bbbz$: 
\[
 u_{t_0}=0,\qquad u_{t_n}=(u_n-u)\prod_{k=0}^{n-1}u_k:=f^{(n)},\quad 
u_{t_{-n}}=\frac{u_{-n}-u}{\prod_{k=1}^{n}u_
{-k}}:=f^{(-n)},\qquad n\ge
1.
\]
The recursion 
operator for the linearised equation is $\cS$. Thus the recursion operator for  
(\ref{ut1}) is
\[
 R=\phi_* \cS\phi_*^{-1}=   u(\cS-1)u(\cS-1)^{-1}\cS\frac{1}{u}, \quad 
\phi_*=\frac{1}{v}\cS-\frac{v_1}{v^2}=u(\cS-1)\frac{1}{v}.
\] 
It generates the hierarchy of symmetries  of the system \eqref{ut1} as follows:
\[
 R^{-1}(f^{(1)})=R(f^{(-1)})=0,\quad f^{(n+1)}=R^{n}(f^{(1)}),\quad 
f^{(-n-1)}=R^{-n}(f^{(-1)}), \qquad n>0.
\]

A minimal decomposition for the recursion operator is given by $R=AB^{-1}$ with
\[
 A=u(\cS-1)u\cS,\qquad B=u(\cS-1)
\]
and difference operators $A,B$ form a preHamiltonian pair. Indeed, the operator 
$B$ is the same as in Example 3 and it is preHamiltonian with the form 
$\omega_B=0$. The difference operator $A=BQ$, where $Q=u\cS$ is invertible 
operator. Thus $A$ is 
also a preHamiltonian operator with the form $\omega_A(a,b)=u(a_1b-b_1 a)$ (see 
Remark \ref{inv}). It is easy to check that $A$ and $B$ are compatible. 

However, the system \eqref{ut1} cannot be cast into a Hamiltonian form for any 
Hamiltonian rational operator $H$. Indeed, the equation 
\begin{eqnarray*}
X_{f^{(1)}}(H)={f^{(1)}}_*H+H{f^{(1)}}_*^{\dagger}
\end{eqnarray*}
 has no 
solutions for $H\in\cQ$, since the order of ${f^{(1)}}_*$ is $(0,1)$ while 
the order of ${f^{(1)}}_*^{\dagger}$ is $(-1,0)$.

\section*{Acknowledgements}
The paper is supported by AVM's EPSRC grant EP/P012655/1 and JPW's EPSRC grant 
EP/P012698/1. Both authors gratefully acknowledge the financial support.
JPW and SC were partially supported by Research in 
Pairs grant no. 41670 from the London Mathematical Society; SC also thanks the 
University 
of Kent for the hospitality received during his visit in July 2017. SC was supported by a Junior 
Fellow award from the Simons Foundation.

\end{document}